\documentclass[letterpaper]{article}
\def\year{2018}
\pdfoutput=1
\usepackage{aaai18}
\usepackage{times}
\usepackage{helvet}
\usepackage{courier}

\usepackage{booktabs, amsmath, amssymb,mathtools, microtype}

\usepackage[shortlabels]{enumitem}
\usepackage[hidelinks]{hyperref}
\hypersetup{breaklinks=true}
\usepackage{cleveref}
\usepackage{aaai_hyperref}

\usepackage{kbordermatrix}

\newcommand{\rank}{\textrm{rank}}

\usepackage{amsthm}
\newtheorem{theorem}{Theorem}
\newtheorem*{theorem*}{Theorem}

\newtheorem{proposition}[theorem]{Proposition}

\usepackage{tikz, wrapfig,graphicx}
\usetikzlibrary{positioning, shapes.geometric, matrix, backgrounds, calc, decorations.pathreplacing,decorations.markings}

\renewcommand*{\le}{\leqslant}
\renewcommand*{\ge}{\geqslant}
\renewcommand*{\epsilon}{\varepsilon}

\newcommand{\pref}{\succcurlyeq}

\newcommand*{\citep}[1]{\cite{#1}}	
\newcommand*{\citet}[1]{\citeauthor{#1} \shortcite{#1}}
\newcommand*{\citeNP}[1]{\citeauthor{#1} \citeyear{#1}}

\newcommand{\boldalpha}{\boldsymbol{\alpha}}
\newcommand{\wPAV}{\ensuremath{\boldsymbol{\alpha}}\textup{-PAV}}

\newcommand\bovermat[2]{%
	\makebox[0pt][l]{$\smash{\overbrace{\phantom{%
					\begin{matrix}#2\end{matrix}}}^{\text{#1}}}$}#2}

\begin{document}

\title{Single-Peakedness and Total Unimodularity: \\ New Polynomial-Time Algorithms for Multi-Winner Elections} 
\author{Dominik Peters\\
	Department of Computer Science \\
	University of Oxford, UK \\
	dominik.peters@cs.ox.ac.uk
}
\copyrightyear{2018}
\maketitle
\global\csname @topnum\endcsname 0
\global\csname @botnum\endcsname 0

\begin{abstract}
The winner determination problems of many attractive multi-winner voting rules are NP-complete. However, they often admit polynomial-time algorithms when restricting inputs to be single-peaked. Commonly, such algorithms employ dynamic programming along the underlying axis. We introduce a new technique: carefully chosen integer linear programming (IP) formulations for certain voting problems admit an LP relaxation which is totally unimodular if preferences are single-peaked, and which thus admits an integral optimal solution. This technique gives efficient algorithms for finding optimal committees under Proportional Approval Voting (PAV) and the Chamberlin--Courant rule with single-peaked preferences, as well as for certain OWA-based rules. For PAV, this is the first technique able to efficiently find an optimal committee when preferences are single-peaked. An advantage of our approach is that no special-purpose algorithm needs to be used to exploit structure in the input preferences: any standard IP solver will terminate in the first iteration if the input is single-peaked, and will continue to work otherwise.
\end{abstract}

\section{Introduction}
\label{sec:intro}

In a departure from classical voting theory, a growing literature from computational social choice has recently studied \emph{multi-winner} voting rules \citep{multiwinner-trends}. Given diverse preferences of a collection of agents, instead of identifying a single best alternative, we are aiming for a (fixed-size) \emph{set} of alternatives that jointly are able to represent the preferences of the agents best. Such procedures are useful in a wide variety of circumstances: obvious examples include the election of a parliament, or of a committee representing the interests of members of an organisation. Other applications can be found in group recommendation systems, or for making decisions about which products or services to offer: Which courses should be offered at a university? Which movies should be presented on an airline entertainment system?

%
%
%

Several attractive rules for such tasks have been designed by researchers in political science (e.g., \citeNP{cha-cou:j:cc}, \citeNP{monroe1995fully}) and more recently by computer scientists \cite{faliszewski2016plurality,faliszewski2016hierarchy,skowron2015owa}. Many of these rules are defined in terms of some objective function: a winning committee is a set of $k$ candidates that maximises this objective. Unsurprisingly, then, the winner determination problems of such rules are typically NP-hard \cite{bou-lu:c:chamberlin-courant}. To tackle the complexity of these problems, approximation algorithms \cite{skowron2015achieving} and fixed-parameter tractability approaches \cite{betzler2013computation,bredereck2015elections} have been developed for these problems, and integer programming formulations have also been designed for them \cite{potthoff1998proportional}.

Another approach 
seeks to exploit underlying structure in the preferences reported by the agents \citep{ELP-trends}. A particularly influential preference restriction is the notion of \emph{single-peaked preferences}, due to \citet{black1948rationale} and \citet{arrow1950difficulty}. In this model, the alternative space has a one-dimensional structure: alternatives are ordered on a left-to-right \emph{axis}; and agents' preferences are monotonically decreasing as we move further away from their \emph{peak} (most-preferred alternative). In particular, we can expect preferences to be structured this way when voting over the value of a numerical quantity (such as a tax rate). While single-peaked preferences were first employed to escape impossibility results in social choice theory \cite{mou:b:axioms}, it also yields positive algorithmic results: Notably, Betzler et al.\ \shortcite{betzler2013computation} showed that Chamberlin--Courant's \shortcite{cha-cou:j:cc} multi-winner rule is efficiently computable for single-peaked preferences. This result can be extended to some other multi-winner voting rules \cite{elkind2015owa}, and to preferences that are single-crossing, or single-peaked on trees \cite{skowron2013complexity,elkind2014characterization,peters2016nicetrees}.

The Chamberlin--Courant rule usually takes as input preferences specified by linear orders, i.e., rankings of the candidates. An alternative way of specifying preferences is by \emph{approval ballots}, where voters submit a \emph{set} of approved candidates. Recent work has uncovered a rich variety of multi-winner voting rules in this framework \cite{kilgour2010approval}.
A particularly well-studied rule is \emph{Proportional Approval Voting} (PAV) due to \citet{Thie95a}.
PAV has particularly nice axiomatic properties in terms of proportional representation \cite{aziz2017justified}, and has recently been axiomatically characterised as the only consistent extension of d'Hondt's proportionality criterion for party list elections to the general setting of approval votes (\citeNP{lackner2017consistent}, see also \citeNP{brill2017apportionment}).
Like for Chamberlin--Courant, the winner determination problem of PAV is NP-complete \cite{aziz2015computational}.
For PAV, however, tractability results for cases of structured preferences have proven somewhat elusive.
\citet{elkind2015structure} have studied this question in depth, but for single-peaked preferences\footnote{For approval preferences, single-peakedness requires that there is a left-to-right ordering of the alternatives so that each approval set forms an interval of this ordering \cite{faliszewski2011shield,elkind2015structure}.}
they were only able to establish that the problem lies in FPT for some natural parameters, and that it becomes polynomial-time solvable for a very restrictive subclass of single-peaked approval preferences.
The difficulty is that standard approaches based on dynamic programming seem to not be powerful enough to capture the \emph{multirepresentation} nature of PAV: Under the PAV rule, voters derive (implicit) utility from potentially \emph{many} committee members, whereas in the case of Chamberlin--Courant, they are only represented by one member of the committee.
This difficulty led \citet{elkind2015structure} to conjecture that evaluating PAV is NP-complete even for single-peaked preferences.

In this paper, we introduce a new technique that allows evaluating a large class of multi-winner rules in polynomial time if the input preferences are single-peaked. This class includes PAV; thus we disconfirm Elkind and Lackner's conjecture and show that PAV does become tractable with single-peaked preferences. Other rules in this class include Chamberlin--Courant, the $t$-Borda rules \citep{paths}, and most OWA-based multi-winner rules \cite{skowron2015owa}.

Our technique is based on designing integer linear programming formulations for these rules, and proving that these formulations are \emph{totally unimodular} in case the input preferences are single-peaked. Total unimodularity is a condition on the derminants of the matrix of coefficients appearing in the constraints of the integer program. Since totally unimodular integer programs are optimally solved by their linear programming relaxations, these rules are polynomial-time computable with such inputs. In fact, since all standard IP solvers first solve the LP relaxation, they will terminate with the correct answer in their first iteration. If the instance is not single-peaked, the IP solver might enter further iterations while solving: our formulations are correct whether or not the input is single-peaked. This makes it easy to implement our algorithms, and stands in contrast to algorithms based on dynamic programming, which are specialised to work only under the assumption that the input is single-peaked. 

Another difference to specialised algorithms (such as the one due to Betzler et al.\ \citeyear{betzler2013computation}) is that we do not need to know an underlying single-peaked axis of the input profile in order to solve the integer program in polynomial time. Dynamic programs that work along the single-peaked axis need a preprocessing stage, where we run a \emph{recognition algorithm} which constructs a single-peaked axis if one exists. Such algorithms have been designed by \citet{escoffier2008single}, \citet{bartholdi1989voting}, and \citet{doignon1994polynomial}, among others. Our algorithms skip this step, again making implementation easier.

\section{Preliminaries}

We write $[n] = \{1,\dots,n\}$.
\smallskip

\noindent
\textbf{Total Unimodularity} \quad
A matrix $A = (a_{ij})_{ij} \in \mathbb Z^{m\times n}$ with $a_{ij} \in \{-1,0,1\}$ is called \emph{totally unimodular} if every square submatrix $B$ of $A$ has $\det B \in \{-1,0,1\}$. The following results are well-known. Proofs and much more about their theory can be found in the textbook by \citet{schrijver1998theory}.

\begin{theorem}
	\label{thm:tum-is-easy}
	Suppose $A \in \mathbb Z^{m\times n}$ is a totally unimodular matrix, $b\in \mathbb Z^m$ is an integral vector of right-hand sides, and $c\in\mathbb Q^n$ is an objective vector. Then the linear program
	\[ \text{max } c^Tx \text{ subject to } Ax \le b  \tag{P} \]
	has an integral optimum solution, which is a vertex of the polyhedron $\{x : Ax \le b\}$. Thus, the integer linear program
	\[ \text{max } c^Tx \text{ subject to } Ax \le b, x \in \mathbb Z^n \tag{IP} \]
	is solved optimally by its linear programming relaxation (P).
\end{theorem}

An optimum solution to (IP) can be found in polynomial time.
We will now state some elementary results about totally unimodular matrices. 

\begin{proposition}
	\label{prop:tum-manipulations}
	If $A$ is totally unimodular, then so is
	\begin{enumerate}[(1),leftmargin=1.6em]
		\item its transpose $A^T$,
		\item the matrix $[A \mid -A]$ obtained from $A$ by appending the negated columns of $A$,
		\item the matrix $[A \mid I]$ where $I$ is the identity matrix,
		\item any matrix obtained from $A$ through permuting or deleting rows or columns.
	\end{enumerate}
\end{proposition}
In particular, from (3) and (4) it follows that appending a unit column $(0,\dots,1,\dots,0)^T$ will not destroy total unimodularity. Further, using these transformations, we can see that \Cref{thm:tum-is-easy} remains true even if we add to (P) constraints giving lower and upper bounds to some variables, if we replace some of the inequality constraints by equality constraints, or change the direction of an inequality.

\begingroup
\setlength{\columnsep}{5pt}
\setlength{\intextsep}{0pt}
\begin{wrapfigure}[5]{r}{0.26\columnwidth}
\hspace{-4pt}
\scalebox{0.65}{
\begin{tikzpicture}
[highlight/.style={line width=11pt, black!35!blue!20}]
\matrix (m) [matrix of math nodes,left delimiter={[},right delimiter={]},ampersand replacement=\&] {
	0 \& 0 \& 1 \& 1 \& 1 \& 0 \\
	1 \& 1 \& 1 \& 0 \& 0 \& 0 \\
	0 \& 0 \& 0 \& 0 \& 1 \& 1 \\
	0 \& 1 \& 1 \& 1 \& 1 \& 0 \\
	0 \& 0 \& 0 \& 1 \& 1 \& 0 \\
};
\begin{scope}[on background layer]
\draw[highlight] (m-1-3.west) -- (m-1-5.east);
\draw[highlight] (m-2-1.west) -- (m-2-3.east);
\draw[highlight] (m-3-5.west) -- (m-3-6.east);
\draw[highlight] (m-4-2.west) -- (m-4-5.east);
\draw[highlight] (m-5-4.west) -- (m-5-5.east);
\end{scope}		
\end{tikzpicture}
}
\end{wrapfigure}
A binary matrix $A = (a_{ij}) \in \mathbb \{0,1\}^{m\times n}$ has the \emph{strong consecutive ones property} if the $1$-entries of each row form a contiguous block, as in the example on the right. A binary matrix has the \emph{consecutive ones property} if its columns can be permuted so that the resulting matrix has the strong consecutive ones property. The key result that will allow us to connect single-peaked preferences to total unimodularity is as follows:

\endgroup

\begin{proposition}
	Every binary matrix with the consecutive ones property is totally unimodular.
\end{proposition}

We remark that by a celebrated result of \citet{seymour1980decomposition}, it is possible to decide in polynomial time whether a given matrix is totally unimodular, though we do not use this fact. 
\smallskip

\noindent
\textbf{Single-Peaked Preferences}\quad
Let $A$ be a finite set of \emph{alternatives}, or \emph{candidates}, and let $m=|A|$.
A \emph{weak order}, or \emph{preference relation}, is a binary relation $\pref$ over $A$ that is complete and transitive. We write $\succ$ and $\sim$ for the strict and indifference parts of $\pref$. A \emph{linear order} is a weak order that, in addition, is antisymmetric, so that $x\sim y$ only if $x = y$. Every preference relation $\pref$ induces a partition of $A$ into indifference classes $A_1,\dots,A_r$ so that $A_1 \succ A_2 \succ \dots \succ A_r$ and $x \sim y$ for all $x,y\in A_t$. We will say that an alternative $a \in A_t$ has \emph{rank} $t$ in the ordering $\pref$ and write $\rank(a) = t$; thus the alternatives of rank 1 are the most-preferred alternatives under $\pref$. Finally, we say that any set of the form $\{x \in A: \rank(x) \ge t\}$ is a \emph{top-inital segment} of $\pref$.

\begingroup
\setlength{\columnsep}{5pt}
\setlength{\intextsep}{2pt}
\begin{wrapfigure}[7]{r}{0.45\columnwidth}
\scalebox{0.6}{
	\hspace{-5.8pt}
\begin{tikzpicture}[yscale=0.6,xscale=0.9]
\def\xmin{1}
\def\xmax{7}
\def\ymin{0}
\def\ymax{6}

\draw[step=1cm,black!20,very thin] (\xmin,\ymin) grid (\xmax,\ymax);

\draw[->] (\xmin -0.3,\ymin) -- (\xmax+0.4,\ymin) node[right] {};
\foreach \x/\xtext in {1/a, 2/b, 3/c, 4/d, 5/e, 6/f, 7/g}
\draw[shift={(\x,\ymin)}] (0pt,2pt) -- (0pt,-2pt) node[below] {$\strut\xtext$};
\foreach \x/\xtext in {1, 2,3,4,5,6}
\node[below] at (\x+0.5,\ymin) {$\strut\lhd$};  

\foreach \x/\y in {4/7,5/6,3/4, 6/5, 2/3, 7/2, 1/1}
\node[fill=blue, circle, inner sep=0.6mm] at (\x,\y-1) {};

\draw[thick,blue] (1,0)--(2,2)--(3,3)--(4,6) -- (5,5)--(6,4)--(7,1);

\foreach \x/\y in {4/4,5/3, 6/2, 3/5, 2/7, 7/1, 1/6}
\node[fill=green!50!black, circle, inner sep=0.6mm] at (\x,\y-1) {};

\draw[thick,green!50!black] (1,5)--(2,6)--(3,4)--(4,3) -- (5,2)--(6,1)--(7,0);  

\foreach \x/\y in {1/2,2/4,3/7, 4/6, 5/5, 6/3, 7/1}
\node[fill=red!50!black, circle, inner sep=0.6mm] at (\x,\y-1) {};

\draw[thick,red!50!black] (1,1)--(2,3)--(3,6)--(4,5) -- (5,4)--(6,2)--(7,0);

\end{tikzpicture}
}
\end{wrapfigure}
Let $\lhd$ be (the strict part of) a linear order over $A$; we call $\lhd$ an \emph{axis}. A linear order $\succ_i$ with most-preferred alternative $p$ (the \emph{peak}) is \emph{single-peaked with respect to~$\lhd$} if for every pair of candidates $a,b\in A$ with $p \lhd b\lhd a$ or $a\lhd b\lhd p$ it holds that $b \succ_i a$. For example, if the alternatives in $A$ correspond to different proposed levels of a tax, and the numbers in $A$ are ordered by $\lhd$ in increasing order, then it is sensible to expect voters' preferences over $A$ to be single-peaked with respect to $\lhd$. 

We need an equivalent definition of single-peakedness:
\begin{proposition}
	\label{prop:sp-equiv}
	A linear order $\succ$ is single-peaked with respect to~$\lhd$ if and only if all top-initial segments of $\succ$ form an interval of~$\lhd$.
\end{proposition}
\noindent
The proof is straightforward, see Figure~\ref{fig:sp-to-c1p} for an illustration. An advantage of this alternative definition is that it allows us to generalise single-peakedness to \emph{weak} orders.
Thus, we will define a weak order $\pref$ to be \emph{single-peaked} with respect to $\lhd$ exactly if all top-initial segments of $\pref$ form an interval of $\lhd$. This concept is often known as `possibly single-peaked' \cite{lackner2014incomplete} because it is equivalent to requiring that the ties in the weak order can be broken in such a way that the resulting linear order is single-peaked.

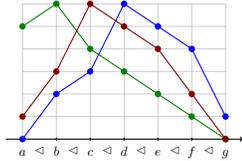
\endgroup
A \emph{profile} $P = (\pref_1,\dots,\pref_n)$ over a set of alternatives~$A$ is a list of weak orders over $A$. Each of the orders represents the preferences of a \emph{voter}; we write $N = [n]$ for the set of voters. The profile will be called \emph{single-peaked} if there exists some axis $\lhd$ over $A$ so that each order $\pref_i$ in $P$ is single-peaked with respect to $\lhd$. 

Based on Proposition~\ref{prop:sp-equiv}, we can characterise single-peakedness of a profile in terms of the consecutive ones property of a certain matrix.
Indeed, a profile is single-peaked if and only if the following matrix $M_{\text{SP}}^P$ has the consecutive ones property: take one column for each alternative, and one row for each top-initial segment $S$ of each voter's preference relation; the row is just the incidence vector of $S$. Figure~\ref{fig:sp-to-c1p} shows an example. Clearly, $M_{\text{SP}}^P$ has the consecutive ones property if and only if the condition of Proposition~\ref{prop:sp-equiv} is satisfied. This construction is due to \citet{bar-tri:j:sp}, see also \citet{fitzsimmons2014single}. Note that, because it is possible to decide whether a matrix has the consecutive ones property in linear time \cite{booth1976testing}, this yields an $O(m^2n)$ algorithm for checking whether a given profile is single-peaked; however, there are faster, more direct algorithms for this task \cite{escoffier2008single,doignon1994polynomial}. 
\smallskip

\begin{figure}[th]
	\centering
	\begin{minipage}{0.18\columnwidth}
		\begin{tabular}{cc}
			\toprule
			$v_1$ & $v_2$ \\
			\midrule
			$b$ & $c$  \\
			$c$ & $d$  \\
			$a$ & $b$  \\
			$d$ & $a$  \\
			\bottomrule
		\end{tabular}
	\end{minipage}
	\quad\:
	\raisebox{-10pt}{\scalebox{1.5}{$\mapsto$}}
	\qquad
	\begin{minipage}{0.56\columnwidth}
		\begin{tikzpicture}
		[decoration=brace,
		highlight/.style={line width=11pt, black!35!blue!20},
		rowlabel/.style={anchor=base, text width=1.5cm}]
		\matrix (m) [matrix of math nodes,left delimiter={[},right delimiter={]}] {
			0 & 1 & 0 & 0 \\
			0 & 1 & 1 & 0 \\
			1 & 1 & 1 & 0 \\
			1 & 1 & 1 & 1 \\
			0 & 0 & 1 & 0 \\
			0 & 0 & 1 & 1 \\
			0 & 1 & 1 & 1 \\
			1 & 1 & 1 & 1 \\
		};
		\node [anchor=base] (a) at ($(m-1-1.north)+(0,5pt)$) {$a$};
		\node [anchor=base] (b) at ($(m-1-2.north)+(0,5pt)$) {$b$};
		\node [anchor=base] (c) at ($(m-1-3.north)+(0,5pt)$) {$c$};
		\node [anchor=base] (d) at ($(m-1-4.north)+(0,5pt)$) {$d$};
		\node [rowlabel] (row1) at ($(m-1-4.east)+(1.5cm,-3pt)$)  {$\{b\}$};
		\node [rowlabel] (row2) at ($(m-2-4.east)+(1.5cm,-3pt)$)  {$\{b,c\}$};
		\node [rowlabel] (row3) at ($(m-3-4.east)+(1.5cm,-3pt)$)  {$\{a,b,c\}$};
		\node [rowlabel] (row4) at ($(m-4-4.east)+(1.5cm,-3pt)$)  {$\{a,b,c,d\}$};
		\node [rowlabel] (row5) at ($(m-5-4.east)+(1.5cm,-3pt)$)  {$\{c\}$};
		\node [rowlabel] (row6) at ($(m-6-4.east)+(1.5cm,-3pt)$)  {$\{c,d\}$};
		\node [rowlabel] (row7) at ($(m-7-4.east)+(1.5cm,-3pt)$)  {$\{b,c,d\}$};
		\node [rowlabel] (row8) at ($(m-8-4.east)+(1.5cm,-3pt)$)  {$\{a,b,c,d\}$};
		\draw[decorate,transform canvas={xshift=-1.3em},thick] ($(m-4-1.south west)+(0,2pt)$) -- node[left=2pt] {$v_1$} (m-1-1.north west);
		\draw[decorate,transform canvas={xshift=-1.3em},thick] (m-8-1.south west) -- node[left=2pt] {$v_2$} ($(m-5-1.north west)+(0,-2pt)$);
		\begin{scope}[on background layer]
		\draw[highlight] (m-1-2.west) -- (m-1-2.east);
		\draw[highlight] (m-2-2.west) -- (m-2-3.east);
		\draw[highlight] (m-3-1.west) -- (m-3-3.east);
		\draw[highlight] (m-4-1.west) -- (m-4-4.east);
		\draw[highlight] (m-5-3.west) -- (m-5-3.east);
		\draw[highlight] (m-6-3.west) -- (m-6-4.east);
		\draw[highlight] (m-7-2.west) -- (m-7-4.east);
		\draw[highlight] (m-8-1.west) -- (m-8-4.east);
		\end{scope}		
		\end{tikzpicture}
	\end{minipage}
	\caption{Translation of single-peakedness of the profile $P$ into the consecutive ones property of the matrix $M_{\text{SP}}^P$: each row corresponds to a top-initial segment.}
	\label{fig:sp-to-c1p}
\end{figure}

\noindent
\textbf{Dichotomous Preferences}\quad
A weak order $\pref$ is \emph{dichotomous} if it partitions $A$ into at most two indifference classes $A_1 \succ A_2$. The alternatives in $A_1$ are said to be \emph{approved} by the voter $\pref$. On dichotomous preferences, the notion of single-peakedness essentially coincides with the consecutive ones property \cite{faliszewski2011shield}: there needs to be an ordering $\lhd$ of the alternative so that each approval set $A_1$ is an interval of $\lhd$. Thus, \citet{elkind2015structure} use the name \emph{Candidate Interval (CI)} for single-peakedness in this context.

\section{Proportional Approval Voting}

In this section, we will consider Proportional Approval Voting (PAV), a multi-winner voting rule defined for dichotomous (approval) preferences. A na\"{i}ve way to form a committee would be to select the $k$ alternatives with highest approval score, but this method tends to ignore minority candidates, and so is not representative \cite{aziz2017justified}. PAV attempts to fix this issue: it is based on maximising a sum of voters' utilities, where a voter $i$'s utility is a concave function of the number of candidates in the committee that $i$ approves of. The rule was first proposed by \citet{Thie95a}. In the general case, a winning committee under PAV is NP-hard to compute \cite{aziz2015computational}, even if each voter approves only 2 candidates and each candidate is approved by only 3 voters.

Let us define PAV formally. Each voter $i$ submits a set $v_i \subseteq C$ of \emph{approved} candidates (or, equivalently, a dichotomous weak order with $v_i \succ_i C \setminus v_i$). We aim to find a good committee $W\subseteq C$ of size $|W| = k$. The intuition behind Proportional Approval Voting (PAV) is that voters are happier with committees that contain more of their approved candidates, but that there are decreasing marginal returns to extra approved candidates in the committee. Concretely, each voter obtains a `utility' of 1 for the first approved candidate in $W$, of $\frac12$ for the second, of $\frac13$ for the third, and so on. 
The objective value of a committee $W\subseteq C$ is thus
\[ \sum_{i\in N} 1 + \frac12 + \frac13 + \cdots + \frac1{|W \cap v_i|}. \]
The choice of harmonic numbers might seem arbitrary, and one can more generally define a rule \wPAV\ where $\boldalpha \in \mathbb R_+^k$ is a non-increasing scoring vector (so $\alpha_i \ge \alpha_j$ when $i \ge j$). This rule gives $W$ the objective value $\sum_{i\in N} \sum_{\ell=1}^{|W \cap v_i|} \alpha_\ell$. Then PAV is just $(1,\frac12,\frac13,\dots,\frac1k)$-PAV. However, the choice of harmonic numbers is the only vector $\boldalpha$ that lets \wPAV\ satisfy an axiom called `extended justified representation' \cite{aziz2017justified}, making this a natural choice after all.

For single-peaked preferences, \citet{elkind2015structure} presented algorithms for finding an optimal PAV committee running in time $O(2^{s}nm)$ and $\operatorname{poly}(d,m,n,k^d)$, where $s$ is the maximum cardinality of the approval sets, and $d$ is the maximum number of voters that approve a given candidate. These algorithms, based on dynamic programming, are efficient if these parameters $s$ or $d$ are small. They also showed, by extending an algorithm of \citet{betzler2013computation}, that \wPAV\ is easy for scoring vectors $\boldalpha = (\alpha_1,\dots, \alpha_r, 0, 0, \dots)$ that are `truncated' in that they only have a constant number of non-zero values at the beginning. Finally, if one imposes further (restrictive) assumptions on the structure of the input preferences, they also found polynomial-time results. However, none of these algorithms could be extended to cover the general case for solving PAV and \wPAV\ for non-truncated scoring vectors. Our method, being very different from dynamic programming, can solve PAV in polynomial time if preferences are single-peaked.

Let us now give and analyse our IP formulation for PAV. This formulation has one binary variable $y_c$ for each candidate $c\in C$, indicating whether candidate $c$ is part of the committee. Constraint (2) requires that the committee contains exactly $k$ candidates. The binary variables $x_{i,\ell}$ indicate whether voter $i\in N$ approves of at least $\ell$ candidates in the committee; this interpretation is implemented by the constraints (3).
\[\arraycolsep=2pt\def\arraystretch{1.8}
\begin{array}{rrclc@{\qquad}lr}
  \text{maximise } & \multicolumn{3}{l}{\displaystyle\sum_{i \in N} \sum_{\ell \in [k]} \alpha_\ell \cdot x_{i,\ell}} && \multicolumn{2}{r}{\text{(PAV-IP)}} \\
  \text{subject to } 
  & \displaystyle\sum_{c\in C} y_{c\hphantom{,\ell}} &=& k &&& \text{(2)} \\
  & \displaystyle\sum_{\ell \in [k]} x_{i,\ell} &\le& \displaystyle\sum_{\mathclap{i \text{ approves } c}} \: y_c && \text{for } i \in N & \text{(3)} \\
  & x_{i,\ell} &\in& \{0,1\} && \text{for } i\in N,\: \ell\in[k] & \text{(4)} \\
  & y_{c\hphantom{,\ell}} &\in& \{0,1\} && \text{for } c\in C
\end{array}
\]

To help intuition, let us explicitly write down an example for PAV with harmonic weights, with 4 candidates $a,b,c,d$, target committee size $k = 2$, and two voters, where voter 1 approves $\{a,b,c\}$, and voter 2 approves $\{c,d\}$. This profile is single-peaked on the axis $a \lhd b \lhd c \lhd d$.
\[\arraycolsep=2pt\def\arraystretch{1.5}
\begin{array}{rrclclr}
\text{maximise } & \multicolumn{3}{l}{(x_{1,1} + \frac12 x_{1,2}) + (x_{2,1} + \frac12 x_{2,2})} && \multicolumn{2}{r}{\text{(PAV-IP')}} \\
\text{subject to } 
&  y_a + y_b + y_c + y_d &=& 2 &&& \text{(2)} \\
& x_{1,1} + x_{1,2} & \le& y_a + y_b + y_c && & \text{(3)} \\
& x_{2,1} + x_{2,2} & \le& y_c + y_d && & \text{(3)} \\
& \multicolumn{3}{c}{\text{all variables binary}}
\end{array}
\]
An optimal committee for this profile is, for example, $W = \{ c,d \}$. An optimum feasible solution would then set $(x_{1,1}, x_{1,2}, x_{2,1}, x_{2,2}) = (1, 0, 1, 1)$. We must set $x_{1,2} = 0$, because voter $1$ only approves of a single candidate in the committee. The resulting objective value is $1 + 0 + 1 + \frac12 = 2.5$, which is precisely the PAV value of the committee $\{c,d\}$.

We now argue formally that our formulation (PAV-ILP) captures the winner determination problem of \wPAV.

\begin{proposition}
	Program \textup{(PAV-IP)} correctly computes an optimal committee according to \wPAV.
\end{proposition}
\begin{proof}
	In any feasible solution of \textup{(PAV-IP)}, the variables $y_c$ encode a committee of size $k$. Fix such a committee $W = \{ c\in C : y_c = 1 \}$. We show that the optimum objective value of a feasible solution with these choices for the $y_c$-variables is the \wPAV-value of this committee. 
	
	Since $\boldalpha \ge 0$, in optimum, as many $x_{i, \ell}$ will be set to 1 as constraint (3) allows. Thus, for each $i$, exactly $|W \cap v_i|$ many variables $x_{i,\ell}$ will be set to 1. Since $\mathbf w$ is non-increasing, wlog, in optimum, these will be variables $x_{i, 1},\dots x_{i, |W \cap v_i|}$. Then the objective value equals the \wPAV-value of $W$.
\end{proof}

The constraint matrix of our example (PAV-ILP') is
\[
A_{\text{PAV-IP'}} =\!\!\! \kbordermatrix{
& x_{1,1} & x_{1,2} & x_{2,1} & x_{2,2} & y_a & y_b & y_c & y_d \\
& 0 & 0 & 0 & 0 & 1 & 1 & 1 & 1 \\
& -1 & -1 & 0 & 0 & 1 & 1 & 1 & 0 \\
& 0 & 0 & -1 & -1 & 0 & 0 & 1 & 1
}
\]
Here we have rewritten constraint (3) as $-x_{1,1} - x_{1,2} + y_a +y_b +y_c \ge 0$.
Note that, because the input profile is single-peaked, the submatrix corresponding to the $y_c$-variables has the consecutive ones property: it is just $M_{\text{SP}}^P$ with an all-$1$s row added. Thus, this submatrix is totally unimodular. The whole constraint matrix is the obtained by adding negative unit columns corresponding to the $x_{i,\ell}$-variables. By Proposition~\ref{prop:tum-manipulations}, adding these columns yields a matrix that is still totally unimodular. Generalising this argument to arbitrary profiles, we get the following.

\begin{proposition}
	\label{prop:pav-tum}
	The constraint matrix of \textup{(PAV-IP)} is totally unimodular when the input preferences are single-peaked.
\end{proposition}
\begin{proof}
	We will use the manipulations allowed by Proposition~\ref{prop:tum-manipulations} liberally. In particular, it allows us to ignore the constraints $0 \le x_{i,\ell}, y_c \le 1$, and to ignore the difference between equality and inequality constraints. Thus, after permuting columns corresponding to variables $x_{i,\ell}$ so that they are sorted by $i$, the constraint matrix of (PAV-IP) is
	\vspace{6pt}
	\[ A_{\text{PAV-IP}} =
	\begin{bmatrix}
	\bovermat{$k$ times}{-I_n & \dots & -I_n} & M_{\text{SP}}^P \\
	\mathbf 0_n & \dots & \mathbf 0_n & \mathbf 1_m
	\end{bmatrix},
	\]
	where $I_n$ is the $n\times n$ identity matrix, and $\mathbf 1_m$ is an all-1s row vector with $m$ entries.
	
	If preferences $P$ are single-peaked, then $M_{\text{SP}}^P$ has the consecutive ones property, and this is also true after appending a row with all-1s. Thus, $\left[\begin{smallmatrix}
	M_{\text{SP}}^P \\ \mathbf 1_m
	\end{smallmatrix}\right]$ is totally unimodular. Applying Proposition~\ref{prop:tum-manipulations} repeatedly to append negations of unit columns, we obtain $A_{\text{PAV-IP}}$, which is thus totally unimodular.
\end{proof}

Using \Cref{thm:tum-is-easy}, that integer programming is easy for totally unimodular programs, we obtain our desired result.

\begin{theorem}
	\wPAV\ can be computed in polynomial time for single-peaked approval preferences.
\end{theorem}

An interesting feature of (PAV-IP) is that the integrality constraints (4) on the variables $x_{i, \ell}$ can be relaxed to just $0 \le x_{i, \ell} \le 1$; this does not change the objective value of the optimum solution. This is because in optimum, the quantity $\sum_{\ell} x_{i,\ell}$ is integral by (3), and it never pays to have one of the $x_{i,\ell}$ to be fractional, because this fractional amount can be shifted to $x_{i, \ell'}$ with $\ell' < \ell$ to get a (weakly) higher value.\footnote{A similar property is used by \citeauthor{bredereck2015elections} (2015, Thm 1).}
This observation might tell us that solving (PAV-IP) is relatively easy as it is ``close'' to being an LP, and this also seems to be true in practice. On the other hand, it is not necessarily beneficial to relax the integrality constraints (4) when passing (PAV-IP) to an IP solver: the presence of integrality constraints might nudge the solver to keep numerical integrality gaps smaller.
Of course, the point of this paper is to give another reason why (PAV-IP) is ``close'' to being an LP, namely when preferences are single-peaked.

\section{Chamberlin--Courant's Rule}

Now we leave the domain of dichotomous preferences, and consider the full generality of profiles of weak orders.
The definition of Chamberlin--Courant's rule \shortcite{cha-cou:j:cc} is based on the notion of having a \emph{representative} in the elected committee: each voter is represented by their favourite candidate in the committee, and voters are happier with more preferred representatives. Let $\mathbf w \in \mathbb N^m$ be a (non-increasing) scoring vector; the standard choice for $\mathbf w$ are Borda scores: $\mathbf w = (m, m-1, \dots, 2, 1)$. Let $P = (\pref_1, \dots, \pref_n)$ be a profile. Then the objective value of a committee $W\subseteq C$ according to Chamberlin--Courant's rule is
\[ \sum_{i\in N} \max \{ w_{\rank_i(c)} : c \in W \}.  \]
Chamberlin--Courant now returns any committee $W\subseteq C$ with $|W| = k$ that maximises this objective.

Betzler et al.\ \shortcite{betzler2013computation} gave a dynamic programming algorithm solving this problem if preferences are single-peaked. We now give an alternative algorithm that has the advantage of being extensible to OWA-based rules, as we will see in the next section.

Chamberlin--Courant can be seen as a (non-metric) facility location problem: each candidate $c\in C$ is a potential facility location, we are allowed to open exactly $k$ facilities, and the distance between customers and facilities are determined through $\mathbf w$. There is a standard integer programming formulation for this problem using binary variables $y_c$, denoting whether $c$ will be opened or not, and variables $x_{i,c}$, denoting whether facility $c$ will service voter $i$. However, this formulation is not totally unimodular, since it encodes the preferences in its objective function rather than in its constraints.

To construct our alternative IP formulation, we need another definition of the Chamberlin--Courant objective function based on maximising a number of \emph{points}. For expositional simplicity, let us take $\mathbf w$ to be Borda scores; other scoring rules can be obtained by weighting the points. Here is another way of thinking about the objective value as defined above: each voter $i$ can earn a point for each \emph{rank} in $i$'s preference order: for every rank $r\in[m]$, $i$ earns the point $x_{i,r}$ if there is a committee member $c\in W$ with $\rank_i(c) \ge r$. Then the number of points obtained in total equals the objective value: if $i$'s favourite committee member is in rank $r$, then $i$ will earn precisely $w_r = m - r + 1$ points, namely the points $x_{i,r}, x_{i, r+1}, \dots, x_{i, m}$. This view suggests the following integer programming formulation, where we put $w_r' = w_r - w_{r-1}$ and $w_1' = w_1$.
\[\arraycolsep=2pt\def\arraystretch{1.8}
\begin{array}{rrclc@{\qquad}l@{\:\:}r}
\text{maximise } & \multicolumn{3}{l}{\displaystyle\sum_{i \in N} \sum_{r \in [m]}  w'_r \cdot x_{i,r}}  && \multicolumn{2}{r}{\text{(CC-IP)}} \\
\text{subject to } 
& \displaystyle\sum_{c\in C} y_{c} &=& k && & \text{(2)} \\
& x_{i,r} &\le& \displaystyle\sum_{\mathclap{c \: : \: \rank(c) \ge r}} \: y_c && \text{for } i \in N,\: r\in[m] & \text{(3)}\\
& x_{i,r} &\in& \{0,1\} && \text{for } i\in N,\: r\in[m] \\
& y_{c\hphantom{,\ell}} &\in& \{0,1\} && \text{for } c\in C
\end{array}
\]

Again let us consider an illustrative instantiation of this program, based on Borda scores, committee size $k = 2$, and the profile from Figure~\ref{fig:sp-to-c1p}, with $b \succ_1 c \succ_1 a \succ_1 d$ and $c \succ_2 d \succ_2 b \succ_2 a$. The resulting program is as follows.
\[\arraycolsep=2pt\def\arraystretch{1.2}
\begin{array}{rrclclr}
\text{maximise } & \multicolumn{6}{l}{(x_{1,1} + x_{1,2} + x_{1,3} + x_{1,4})}
\\ &  \multicolumn{6}{l}{\hspace{1.73cm}{}+ (x_{2,1} + x_{2,2} + x_{2,3} + x_{2,4})} \\
\text{subject to } 
& y_a + y_b + y_c + y_d &=& 2 && & \text{(2)} \\
& x_{1,1} &\le&  y_b && & \text{(3)}\\
& x_{1,2} &\le&  y_b + y_c && & \text{(3)}\\
& x_{1,3} &\le&  y_a + y_b + y_c && & \text{(3)}\\
& x_{1,4} &\le&  y_a + y_b + y_c + y_d\: && & \text{(3)}\\
& x_{2,1} &\le&  y_c && & \text{(3)}\\
& x_{2,2} &\le&  y_c + y_d && & \text{(3)}\\
& x_{2,3} &\le&  y_b + y_c + y_d && & \text{(3)}\\
& x_{2,4} &\le&  y_a + y_b + y_c + y_d && & \text{(3)}\\
& \multicolumn{3}{c}{\text{all variables binary}}
\end{array}
\]
To understand the program, it is convenient to rephrase it in logical language: we can interpret ``$\le$'' as ``only if'', and ``$+$'' as ``or'' in constraints (3). For example, the second constraint of type (3) can be read as $x_{1,2}$ only if the committee contains either $b$ or $c$.

Clearly, for this profile, the committee $\{b,d\}$ is optimal, which allows us to set all the $x_{i,\ell}$-variables to 1, yielding objective value 8. If we were to take the committee size $k = 1$ (in which case Chamberlin--Courant is the same as the Borda count), then $\{c\}$ is the optimum committee. That committee forces us to set $x_{1,1} = 0$, but allows us to set all the others to $1$, yielding objective value 7, which is $c$'s Borda score. For general profiles, we can see correctness of the formulation as follows.

\begin{proposition}
	Program \textup{(CC-IP)} correctly computes an optimal committee according to $\mathbf w$-Chamberlin--Courant.
\end{proposition}
\begin{proof}
	In any feasible solution of \textup{(CC-IP)}, the variables $y_c$ encode a committee of size $k$. Fix such a committee $W = \{ c\in C : y_c = 1 \}$. We show that the optimum objective value of a feasible solution with these choices for the $y_c$-variables is the objective value of this committee according to $\mathbf w$-CC. 
	
	Since $\mathbf w' \ge 0$, in optimum, every $x_{i,r}$ will be set to $1$ if constraint (3) allows this. This is the case iff there is a committee member $c\in W$ with $\rank_i(c) \ge r$, i.e., iff the `point' $x_{i,r}$ is earned as described above. Thus, the objectives of (CC-IP) and Chamberlin--Courant coincide.
\end{proof}

Similarly to the case of PAV, the constraint matrix of our example instance can be obtained by taking the matrix in Figure~\ref{fig:sp-to-c1p}, appending a row of all-1s for constraint (2), and then adding negative unit columns.
By Proposition~\ref{prop:tum-manipulations}, the resulting matrix is totally unimodular, because the matrix $M_{\text{SP}}^P$ that we started out with has the consecutive ones property. More generally, the proof is as follows.

\begin{proposition}
	The constraint matrix of \textup{(CC-IP)} is totally unimodular when the input preferences are single-peaked.
\end{proposition}
\begin{proof}
	After similar simplification as in Proposition~\ref{prop:pav-tum} we see that the constraint matrix of (CC-IP) is
	\[ A_{\text{CC-IP}} =
	\begin{bmatrix}
	-I_{nm} & M_{\text{SP}}^P \\
	0 & \mathbf 1_m
	\end{bmatrix}.
	\]
	Again, if preferences $P$ are single-peaked, then $M_{\text{SP}}^P$ has the consecutive ones property, and this is also true after appending a row with all-1s. Thus, $\left[\begin{smallmatrix}
	M_{\text{SP}}^P \\ \mathbf 1_m
	\end{smallmatrix}\right]$ is totally unimodular. Applying Proposition~\ref{prop:tum-manipulations} repeatedly to append unit columns, we obtain $A_{\text{CC-IP}}$, which is thus totally unimodular.
\end{proof}

\begin{theorem}
	Chamberlin--Courant with score vector $\mathbf w$ can be solved in polynomial time for single-peaked preferences.
\end{theorem}

\section{OWA-based Rules}

When using Chamberlin--Courant, each voter is represented by exactly one committee member. In many applications, we may instead seek \emph{multirepresentation} \cite{skowron2015owa}: e.g., you might watch several of the movies offered by an inflight-entertainment system, and thus should be represented by several committee members. In such scenarios, Chamberlin--Courant might design a suboptimal committee; Skowron et al.\ \shortcite{skowron2015owa} introduce \emph{OWA-based} multi-winner rules as a more flexible alternative (see also \citeNP{faliszewski2016hierarchy}). OWA-based rules generalise both Chamberlin--Courant and PAV.

Another way of motivating OWA-based rules, as defined below, is by examining the contrast between the $k$-Borda rule and Chamberlin--Courant, a line of thought explored by \citet{paths}. The $k$-Borda rule returns the committee formed of the $k$-candidates whose Borda scores are highest. This rule is suitable if we aim to identify a collection of `excellent candidates', such as when shortlisting candidates for job interviews, or when finding finalists of various competitions, based on evaluations provided by judges. On the other hand, Chamberlin--Courants identifies committees that are \emph{diverse}, in that they represent as many of the voters as possible. In many real-life applications, we desire a compromise between these two extremes; for example, we might want our shortlist of excellent candidates to also be somewhat diverse. \citet{paths} propose the $t$-Borda rules, $t = 1,\dots, k$, as a family of rules that interpolate between $k$-Borda and Chamberlin--Courant. Under the $t$-Borda rule, a voter $i$'s utility is measured by the sum of the Borda scores of $i$'s top $t$ most-preferred members of the committee. Thus, Chamberlin--Courant is identical to $1$-Borda, and $k$-Borda is identical to, well, $k$-Borda.

Let us now formally define OWA-based rules.
Given a vector $\mathbf x \in \mathbb R^k$, a weight vector $\boldsymbol \alpha\in \mathbb R^k$ defines an \emph{ordered weighted average} (OWA) operator as follows: first, sort the entries of $\mathbf x$ into non-increasing order, so that $x_{\sigma(1)} \ge \dots \ge x_{\sigma(k)}$; second, apply the weights: the ordered weighted average of $\mathbf x$ with weights $\boldsymbol \alpha$ is given by $\boldalpha(\mathbf x) := \sum_{i=1}^k \alpha_ix_{\sigma(i)}$. For example, $\boldsymbol \alpha = (1,0,\dots,0)$ gives the maximum, and $\boldsymbol \alpha = (1,1,\dots,1)$ gives the sum of the numbers in $\mathbf x$.

Now, a scoring vector $\mathbf w \in \mathbb N^m$ and an OWA $\boldalpha$ define an OWA-based multi-winner rule as follows: Given a profile~$P$, the rule outputs a committee $W = \{c_1, \dots, c_k\}$ that maximises the objective value
\[ \sum_{i\in N} \boldalpha(w_{c_1}, \dots,w_{c_k}). \]
Thus, choosing $\boldalpha = (1,0,\dots,0)$ gives us $\mathbf w$-Chamberlin--Courant as a special case. Choosing $\boldalpha = (1,1,0,\dots,0)$ gives us 2-Borda, the analogue of Chamberlin--Courant where voters are represented by their favourite \emph{two} members of the committee. The OWA-based rules with $\boldalpha = (1,\frac12,\dots,\frac1k)$ and $\mathbf w = (1,0,\dots,0)$ gives us PAV, when given dichotomous preferences as input.  Thus, OWA-based rules generalise both Chamberlin--Courant and PAV, and it turns out that we can apply our method to these rules by merging the ideas of (PAV-IP) and (CC-IP). However, our formulation is only valid for \emph{non-increasing} OWA vectors with $\alpha_i \ge \alpha_j$ whenever $i \ge j$. For example, this excludes the rule where voters are represented by their \emph{least}-favourite committee member.\footnote{Still, this case is also efficiently solvable in the single-peaked case: note that a voter's least-favourite committee members will be either the left-most or the right-most member of the committee; thus it suffices to consider committees of size 2. This idea can be extended to OWA operators $\boldalpha = (0,\dots,0,\alpha_{k-c},\dots,\alpha_k)$ that are zero except for constantly many values at the end.}

In the IP, we use variables $x_{i,\ell, r}$ indicating whether voter $i\in N$ ranks at least $\ell$ candidates in the committee in rank $r$ or higher. We again put $w_r' = w_r - w_{r-1}$ and $w_1' = w_1$.
\[\arraycolsep=2pt\def\arraystretch{1.8}
\begin{array}{rrclc@{\quad}lr}
\text{maximise} & \multicolumn{6}{l}{\displaystyle\sum_{i \in N\hphantom{]}} \sum_{\ell \in [k]} \sum_{r\in [m]} \alpha_\ell \cdot w_r' \cdot x_{i,\ell,r} \hfill \text{(OWA-IP)}} \\
\text{subject to} 
& \displaystyle\sum_{c\in C} y_{c\hphantom{,\ell,r}} &=& k &&& \text{(2)} \\
& \displaystyle\sum_{\ell \in [k]} x_{i,\ell, r} &\le& \displaystyle\sum_{\mathclap{c \: : \: \rank(c) \ge r}} \: y_c && \text{for } i \in N,\: r\in[m] & \text{(3)}\\
& x_{i,\ell,r} &\in& \{0,1\} && \text{for } i\in N,\: \ell, \: r & \text{(4)} \\
& y_{c\hphantom{,\ell,r}} &\in& \{0,1\} && \text{for } c\in C
\end{array}
\]

\begin{proposition}
	If $\boldalpha$ and $\mathbf w$ are non-increasing, \textup{(OWA-IP)} correctly computes an optimal committee according to the OWA-based rule based on $\boldalpha$ and $\mathbf w$.
\end{proposition}
\begin{proof}[Proof sketch]
	Similarly to previous arguments, in optimum, we will have $x_{i,\ell,r} = 1$ if and only if the committee $W = \{ c\in C : y_c = 1 \}$ contains at least $\ell$ candidates that voter $i$ ranks in rank $r$ or above. Thus, the objective value of (OWA-IP) agrees with the defined objective of the OWA-based rule.
\end{proof}

The following property is proved very similarly to before: the constraint matrix is obtained from $M_{\text{SP}}^P$ by appending unit columns.

\begin{proposition}
	The constraint matrix of \textup{(OWA-IP)} is totally unimodular when  input preferences are single-peaked.
\end{proposition}

\begin{theorem}
	\label{thm:OWA}
	OWA-based rules with non-increasing OWA operator can be solved in polynomial time for single-peaked preferences.
\end{theorem}

Since Chamberlin--Courant and PAV are special cases of OWA-based rules, Theorem~\ref{thm:OWA} implies our previous results.

\section{Some Extensions}
\noindent
\textbf{More than Single-Peakedness.}
Our polynomial-time results apply to a slightly larger class than just single-peaked profiles: they also apply when $M^P_{\text{SP}}$ (with an all-1s row appended) is totally unimodular but does not necessarily have the consecutive ones property. It can be shown that this is the case whenever $P$ contains only two distinct voters, or, more generally, when the set of all top-initial segments of $P$ can also be induced by a two-voter profile. Together with single-peaked profiles, we conjecture that these classes of profiles are precisely the profiles for which the relevant constraint matrices are totally unimodular.
\smallskip

\noindent
\textbf{Single-Peaked on a Circle.}
\citet{peters2017spoc} introduce the notion of preferences single-peaked on a circle, a generalisation of single-peakedness. They show, by adapting our total unimodularity technique, that all the multiwinner rules we consider remain tractable even for their larger preference restriction. Notably, their argument requires running a recognition algorithm first, and then rearranging our ILP formulations in order to obtain total unimodularity.
\smallskip

\noindent
\textbf{Egalitarian versions.} We can obtain egalitarian versions of the multi-winner rules that we have discussed by replacing the sum over $N$ by a minimum in their objective values \cite{betzler2013computation}. For PAV and Chamberlin--Courant, our IP formulations can easily be adapted to answer the question ``is there a committee with egalitarian objective value $\ge L$?'' while preserving total unimodularity in the case of single-peaked preferences. An optimum committee can then be found by a binary search on $L$. However, it is unclear how this can be achieved for OWA-based rules. It is also unclear how to handle other utility aggregation operators such as leximin (see \citeNP{elkind2015owa}). \smallskip

\noindent
\textbf{PAV and Voter Intervals.} \citet{elkind2015structure} define an analogue of single-crossingness for dichotomous preferences called \emph{voter interval} (VI), which requires the \emph{transpose} of $M^P_{\text{SP}}$ to have the consecutive ones property. 
As for CI, they conjectured that PAV remains hard on VI preferences. We could not solve this problem using our method: the constraint $\sum_{c\in C} y_c = k$ of (PAV-IP) destroys total unimodularity. \smallskip

\noindent
\textbf{Monroe's rule.} \citet{monroe1995fully} proposed a version of Chamberlin--Courant in which each committee member should represent (roughly) the same number of voters. Since this rule (in a certain version) remains NP-complete to evaluate for single-peaked preferences (Betzler et al.\ \citeyear{betzler2013computation}), our approach will likely not extend to Monroe's rule.
\smallskip

\noindent
\textbf{Dodgson's rule.} An alternative is a \emph{Dodgson winner} if it can be made a Condorcet winner using a minimum number of swaps of adjacent alternatives. This number of swaps is the \emph{Dodgson score} of an alternative. \citet{bartholdi1989voting} give an IP formulation for this problem, which is also used in the treatment of \citet{caragiannis2009approximability}. Sadly, while `most' of the constraint matrix is again identical to $M^P_{\text{SP}}$, some extra constraints (saying that the swaps in each vote should only count once) destroy total unimodularity, so our method cannot be employed for this formulation. Note that while \citet{brandt2015bypassing} give an efficient algorithm for finding a Dodgson \emph{winner} in the case of single-peaked preferences, the problem of efficiently calculating \emph{scores} appears to be open and non-trivial. Maybe there is an alternative IP formulation that can be made to work using our approach. \smallskip

\noindent
\textbf{Kemeny's rule.} \citet{conitzer2006kemeny} present several IP formulations for Kemeny's rule. The poly-size formulation they give involves constraints enforcing transitivity of the Kemeny ranking; these constraints are not totally unimodular. In any case, most strategies for calculating Kemeny's rule first calculate all pairwise majority margins; we might as well check for transitivity at this stage -- trying to use fancy total unimodularity is unnecessary. \smallskip

\section{Conclusions}

We presented a new algorithmic technique that can evaluate some popular multi-winner voting rules in polynomial time when preferences are single-peaked. Interestingly, this approach works even if we have not checked in advance whether single-peakedness applies. In future work, it will be interesting to see how these formulations perform in practice, and to see whether other problems admit formulations of this type.
\smallskip

\noindent
\textbf{Acknowledgements.} 
I thank Martin Lackner and Piotr Skowron for helpful discussions and insights. I am supported by EPSRC and through ERC grant 639945 (ACCORD).

\bibliographystyle{plainnat}

\end{document}